\documentclass[11pt,reqno]{amsart}
\usepackage{amssymb,mathrsfs,graphicx}
\usepackage{ifthen}
\usepackage{colortbl}
\definecolor{black}{rgb}{0.0, 0.0, 0.0}
\definecolor{red}{rgb}{1.0, 0.5, 0.5}

\usepackage{ifthen} 
\provideboolean{shownotes} 
\setboolean{shownotes}{true} 
\newcommand{\margnote}[1]{
\ifthenelse{\boolean{shownotes}}%
{\marginpar{\raggedright\tiny\texttt{#1}}}%
{}%
}
\newcommand{\hole}[1]{
\ifthenelse{\boolean{shownotes}}%
{\begin{center} \fbox{ \rule {.25cm}{0cm} \rule[-.1cm]{0cm}{.4cm}
\parbox{.85\textwidth}{\begin{center} \texttt{#1}\end{center}} \rule
{.25cm}{0cm}}\end{center}} {} }

\title[C-S model with singular communication weight]{Local well-posedness of the generalized Cucker-Smale model}

\author[Carrillo]{Jos\'{e} A. Carrillo}
\address[Jos\'{e} A. Carrillo]{\newline Department of Mathematics
    \newline Imperial College London, London SW7 2AZ, United Kingdom}
\email{carrillo@imperial.ac.uk}

\author[Choi]{Young-Pil Choi}
\address[Young-Pil Choi]{\newline Department of Mathematics
    \newline Imperial College London, London SW7 2AZ, United Kingdom}
\email{young-pil.choi@imperial.ac.uk}

\author[Hauray]{Maxime Hauray}
\address[Maxime Hauray]{\newline Centre de Math\'ematiques et Informatique (CMI), \newline
    Universit\'e de Provence, Technop\^ole Ch\^ateau-Gombert, Marseille, France}
\email{maxime.hauray@univ-amu.fr}

\numberwithin{equation}{section}

\newtheorem{theorem}{Theorem}[section]
\newtheorem{lemma}{Lemma}[section]

\newtheorem{proposition}{Proposition}[section]
\newtheorem{remark}{Remark}[section]
\newtheorem{definition}{Definition}[section]

\newcommand{\R}{\mathbb R}
\newcommand{\pp}{\mathcal{P}}
\newcommand{\bbr}{\mathbb R}

\newcommand{\mt}{\mathcal{T}}
\newcommand{\e}{\varepsilon}

\def\charf {\mbox{{\text 1}\kern-.30em {\text l}}}

\begin{document}
\allowdisplaybreaks

\date{\today}



\begin{abstract}
In this paper, we study the local well-posedness of two types of generalized Cucker-Smale (in short C-S) flocking models. We consider two different communication weights, singular and regular ones, with nonlinear coupling velocities $v|v|^{\beta-2}$ for $\beta > \frac{3-d}{2}$. For the singular communication weight, we choose $\psi^1(x) = 1/|x|^{\alpha}$ with $\alpha \in (0,d-1)$ and $\beta \geq 2$ in dimension $d > 1$. For the regular case, we select $\psi^2(x) \geq 0$ belonging to $(L_{loc}^\infty \cap \mbox{Lip}_{loc})(\R^d)$ and $\beta \in (\frac{3-d}{2},2)$. We also remark the various dynamics of C-S particle system for these communication weights when $\beta \in (0,3)$.
\end{abstract}

\maketitle \centerline{\date}

\tableofcontents

%
%
%
%
\section{Introduction}
In the last years, collective problems as a dynamic feature of autonomous agents are active recent subjects in many different disciplines such as statistical physics, mathematics, biology, and control theory, etc., due to its engineering, physical, and biological applications \cite{Bir,CDF,CKFL,DM,LP,PE,PEG,TT}. Among various mathematical models describing the interactions between the individuals, our interest in the paper lies the flocking model which was introduced by Cucker-Smale \cite{CS1,CS2}. Cucker-Smale(in short C-S) model is a type of Newton particle model motivated by the work of Vicsek et. al. \cite{Vic}, and rigorous asymptotic flocking estimates depending on the decay rate of regular communication weight were provided in \cite{CS2}. Later, these estimates were improved and refined in the literature \cite{CFRT,HL,HT}. 

Despite of its novelty of C-S  model for flocking dynamics, there are several drawbacks in real applications. Among them, our study is dealing with two issues; collision avoidance between individuals and general coupling for velocities of them. For the real applications of C-S model, e.g., unmanned aerial vehicle, it is important to avoid any collision. However the original C-S model \cite{CS1,CS2} does not take into account this, and as a consequence, there are some studies to prevent the collisions by adding new forcing terms to control the distance between the individuals \cite{CD,PKH} or considering a modified communication weight \cite{ACHL}. Concerning the general coupling for velocities, the original C-S model has a linear coupling for velocities, but there is no specific physical reason for this. Thus we will consider the nonlinear velocity coupling which is averaged over the strength of the relative speed with a certain exponent.
 
More precisely, let $f = f(x,v,t)$ be the one-particle distribution function at a spatial domain
$x, v \in \bbr^d$ at time $t$ in dimension $d >1$.
In this situation the density function $f$ is determined by
\begin{equation}\label{k-CS-si}
\partial_t f + v \cdot \nabla_x f
 + \nabla_v \cdot \big[F_i(f)f\big] = 0,
  \quad (x, v) \in \bbr^d \times \bbr^d,~~t > 0,
\end{equation}
subject to initial data:
\begin{equation}\label{ini-k-CS}
f(x,v,0) =: f^0(x,v), \quad (x, v) \in \bbr^d \times \bbr^d,
\end{equation}
where $F_i$ denotes the alignment force between particles:
\begin{equation*}
F_i(f)(x,v,t) := -\int_{\bbr^d \times \bbr^d}
\psi^i(x-y)\nabla_v \phi(v-w)f(y,w,t)dydw, \quad i=1,2.
\end{equation*}
Here the potential function $\phi(v)$ for the coupling of velocity is given by 
\begin{equation*}
\phi(v) = \frac1\beta|v|^\beta, \quad \beta > \frac{3-d}{2}
\end{equation*}
Note that if $\beta=2$, then $\nabla\phi(v) = v$, and the velocity coupling is the same with Cucker-Smale's one. For the communication weight $\psi^i$, we take the following two different forms:
\begin{equation}\label{condi:psi}
\psi^1(x) := \frac{1}{|x|^\alpha}, \quad \alpha \in (0,d-1) \quad \mbox{and} \quad 0 \leq \psi^2(x) = \psi^2(-x) \in \left(L_{loc}^\infty \cap \mbox{Lip}_{loc}\right)(\R^d).
\end{equation}
Here $\psi^2 \in \mbox{Lip}_{loc}(\R^d)$ implies that for any compact set $K \subset \R^d$, there is some constant $L_K > 0$ such that
\[
|\psi^2(x) - \psi^2(y)| \leq L_K|x-y|, \quad x,y \in K.
\]
In order to emphasize on the fact that $\beta$ depends on the two different communication weights $\psi^i$, we denote the exponent $\beta_i$ of the coupling of velocity corresponding to $\psi^i$ for $i=1,2$. 

The purpose of the paper is to establish a local existence of unique weak solutions to the kinetic C-S equations \eqref{k-CS-si}-\eqref{ini-k-CS}. For the singular communication weight $F_1$, the singularity with respect to the position is allowed up to the one of Newtonian interactions, and the exponent $\beta_1$ of the potential function is greater than or equals to original Cucker-Smale's one, i.e., $\alpha \in (0,d-1)$ and $\beta_1 \geq 2$. In the other case, the communication weight $F_2$ is sufficiently regular, but the coupling for velocity has a singularity of the order of $\beta_2 \in \left( \frac{3-d}{2}, 2\right)$. For both cases, the general framework for the well-posedness of solutions can not be applied due to the singularity either in position or velocity, and it has not been addressed in the literature to the best of our knowledge. We adapted the idea from our recent work \cite{CCH} to overcome these difficulties. On the other hand, in case of enough regularity of the communication weight and coupling for velocity in the forcing term, the well-posedness for these kinetic equations describing the collective behaviour is included in \cite{CCR}. We also refer to \cite{HJ1,HJ2} for the issue of the mean-field derivation of Vlasov-type equations with singular kernel.

The rest of this paper is organized as follows. In Section \ref{sec2}, we briefly provide definition and properties of Wasserstein distances, and state our main results on well-posedness. Section 3 is devoted to give the details of the proof of a unique weak solution to the system \eqref{k-CS-si} for each case. Our strategy is first to construct approximate solutions, and obtain the uniform bounds of approximate solutions with respect to the regularization. Then we finally show that the approximate solutions are Cauchy sequence, and let the parameter of regularization tend to zero to have the existence of the weak solutions. Finally, in the last section we remark that the dynamics of the generalized C-S particles system can lead to collisions in finite time and we discussed the rate and the conditions under which this can happen.\newline

\noindent {\bf Notations:} $|\cdot|$ denotes the Euclidean distance, and $\mathcal{P}_p(\R^d)$ stands for the set of probability measures with bounded moments of order $p \in [1,\infty)$. For notational simplicity, we also use the following notations throughout the paper: For $1 \leq p \leq \infty$,
\[
\|g\|_{L^p} := \|g\|_{L^p(U)} \quad \mbox{where $U$ can be either $\R^d$ or $\R^d \times \R^d$},
\]
\[
\|g\|_{L^1 \cap L^p} := \|g\|_{L^1} + \|g\|_{L^p}\,, \quad \mbox{and} \quad \|g\| := \|g\|_{L^{\infty}(0,T;L^{1} \cap
L^{p})}.
\]

%
%
%
%
\section{Preliminaries and main results}\label{sec2}

\subsection{Mathematical tools}
In this
part, we present several definition and properties of Wasserstein
distances that will be mainly used in our arguments for the well-posedness.

\begin{definition}(Wasserstein p-distance) \label{defdp}
Let $\rho_1,~ \rho_2$ be two Borel probability measures on
$\bbr^d$. Then the Euclidean Wasserstein distance of order $1\leq
p<\infty$ between $\rho_1$ and $\rho_2$ is defined as
\begin{equation*}\label{wasser-p-dis}
d_p(\rho_1,\rho_2) := \inf_{\gamma} \left( \int_{\bbr^d \times
\bbr^d} |x-y|^p \, d\gamma(x,y) \right)^{1/p},
\end{equation*}
and, for $p=\infty$ (this is the limiting case, as $p \to
\infty$),
\begin{equation*}\label{wasser-inf-dis}
d_{\infty}(\rho_1,\rho_2) := \inf_{\gamma} \left(\sup_{(x,y) \in \text{supp}(\gamma)} | x- y | \right),
\end{equation*}
where the infimum runs over all transference plans, i.e., all
probability measures $\gamma$ on $\bbr^d \times \bbr^d$ with
marginals $\rho_1$ and $\rho_2$ respectively,
\[
\int_{\bbr^d \times \bbr^d} \phi(x) d\gamma(x,y) = \int_{\bbr^d}
\phi(x) \rho_1(x) dx,
\]
and
\[
\int_{\bbr^d \times \bbr^d} \phi(y) d\gamma(x,y) = \int_{\bbr^d}
\phi(y) \rho_2(y) dy,
\]
for all $\phi \in \mathcal{C}_b(\bbr^d)$.
\end{definition}
Note that $\mathcal{P}_p(\R^d), 1\leq p < \infty$ is a
complete metric space endowed with the $p$-Wassertein distance
$d_p$, see \cite{Vil}. We refer to \cite{GS,MC} for more
details in the case of the $d_\infty$ distance.

In particular $p=1$, Wasserstein-1 distance $d_1$ is equivalent to the bounded Lipschitz distance which is also called {\it Monge-Kantorovich-Rubinstein distance}:
\[
d_1(\rho_1,\rho_2) = \sup\left\{ \int_{\R^d} \varphi(\xi)(\rho_1(\xi) - \rho_2(\xi))d\xi \Big| \varphi \in \mbox{Lip}(\R^d), \mbox{ Lip}(\varphi) \leq 1\right\},
\]
where Lip($\R^d$) denotes the set of Lipschitz functions on $\R^d$, and Lip($\varphi$) the Lipschitz constant of a function $\varphi$.
We also remind the definition of the push-forward of a measure by
a mapping in order to give the relation between Wasserstein
distances and optimal transportation.

\begin{definition}
Let $\rho_1$ be a Borel measure on $\bbr^d$ and $\mathcal{T} :
\bbr^d \to \bbr^d$ be a measurable mapping. Then the push-forward
of $\rho_1$ by $\mathcal{T}$ is the measure $\rho_2$ defined by
\[
\rho_2(B) = \rho_1(\mathcal{T}^{-1}(B)) \quad \mbox{for} \quad B
\subset \bbr^d,
\]
and denoted as $\rho_2 = \mathcal{T} \# \rho_1$.
\end{definition}

We recall in the next proposition some classical properties, which proofs may be found in \cite{Vil}.

\begin{proposition}\label{prop-proper}
(i) The definition of $\rho_2 = \mathcal{T} \# \rho_1$ is equivalent
to
$$
\int_{\bbr^d} \phi(x)\, d\rho_2(x) = \int_{\bbr^d} \phi(\mathcal
T(x))\, d\rho_1(x)
$$
for all $\phi\in \mathcal{C}_b(\R^d)$. Given a probability measure with
bounded $p$-th moment $\rho_0$, consider two measurable mappings
$X_1,X_2 : \bbr^d \to \bbr^d$, then the following inequality
holds.
\[
d_p^p(X_1 \# \rho_0, X_2 \# \rho_0) \leq \int_{\bbr^d \times
\bbr^d} |x-y|^p d\gamma(x,y) = \int_{\bbr^d} | X_1(x) - X_2(x)|^p
d\rho_0(x).
\]
Here, we used as transference plan $\gamma = (X_1 \times
X_2)\#\rho_0$ in Definition \ref{defdp}. \newline

\noindent (ii) Given $\{\rho_k\}_{k=1}^{N}$ and $\rho$ in $\pp_1(\R^d)$, the followings are equivalent:
\begin{itemize}
\item $d_1(\rho_k,\rho) \to 0$ as $k \to +\infty$.
\item $\rho_k$ converges to $\rho$ weakly-* as measures and
\[
\int_{\R^d}|\xi|\rho_k(\xi) d\xi \to \int_{\R^d} |\xi| \rho(\xi) d\xi, \quad \mbox{as} \quad n \to +\infty.
\]
\end{itemize}
\end{proposition}
Finally, we recall {\it a priori} energy estimates of kinetic Cucker-Smale model.
\begin{lemma}\label{k-energy-est}
Let $f$ be any smooth solutions to the system \eqref{k-CS-si}. Then we have
\begin{align*}
\begin{aligned}
&(i) \,\,\,\, \frac{d}{dt} \int_{\bbr^d \times \bbr^d} f dx dv = 0, \quad \frac{d}{dt} \int_{\bbr^d \times \bbr^d} v f dx dv = 0, \cr
&(ii) \,\, \frac{1}{2}\frac{d}{dt} \int_{\bbr^d \times \bbr^d} |v|^2 f dx dv = -\frac{1}{2}\int_{\bbr^{2d} \times \bbr^{2d}} \psi^i(x-y) | v - w|^{\beta_i} f(x,v,t) f(y,w,t) dxdydv dw.
\end{aligned}
\end{align*}
\end{lemma}
%
%
%
%

\subsection{Main results}
In this part, we introduce the notion of weak solution and stability in our frameworks, and state our main results on well-posedness.

We first present two frameworks depending on singularity of communication weight and strength of velocity coupling.
\begin{itemize}
\item {\it Framework A (Singular communication weight and high strength of velocity coupling)}: For the force term $F_1$, initial data $f^0$ has a compact support in velocity belongs to $L^p$ for some $p$. The exponents $\alpha$ and $\beta_1$, appearing in the definition of $F_1$, and $p$ should satisfy
\begin{equation*}
(\alpha + 1)p' < d \quad \mbox{and} \quad \beta_1 \geq 2, \quad \mbox{respectively}.
\end{equation*}
where $p'$ is a conjugate of $p$, i.e., $p' := p/(p - 1)$. 
\item {\it Framework B (Low strength of velocity coupling)}: For the force term $F_2$, initial data $f^0$ has a compact support in position and velocity, belongs to $L^p$ for some $p$, and the exponent $\beta_2$ of nonlinear velocity coupling satisfies $\beta_2 \in \left(\frac{3-d}{2}, 2\right)$. Here the exponent $\beta_2$ and $p$ should satisfy
\[
(3-2\beta_2)p' < d \,\, \mbox{ for } \,\, \beta_2 \in \left( \frac{3-d}{2}, 1\right) \quad \mbox{and} \quad (2-\beta_2)p' < d \,\, \mbox{ for } \,\, \beta_2 \in (1,2).
\]
\end{itemize}

\begin{definition}\label{def-ext} For a given $T \in (0,\infty)$, $f$ is a weak solution of \eqref{k-CS-si} on the time-interval $[0,T)$ if and only if the following condition are satisfied:
\begin{enumerate}
\item
$ f\in L^{\infty}(0,T; (L_+^1 \cap L^p)(\R^d \times\R^d)) \cap \mathcal{C}([0,T],\pp_1(\R^d \times \R^d)), $
\item For all $\Psi \in \mathcal{C}_c^{\infty}(\R^d \times \R^d \times [0,T])$,
\begin{align*}
\begin{aligned}
&\int_{\R^d \times \R^d}f(x,v,T)\Psi(x,v,T) dx dv - \int_0^T \int_{\R^d \times \R^d}
f\big(\partial_t\Psi + \nabla_x \Psi \cdot v
+ \nabla_v \Psi \cdot F_i(f) \big)dxdv dt
 \cr
&\qquad \qquad \qquad  = \int_{\R^d\times\R^d}
 f^0(x,v) \Psi^0(x,v)dxdv.
\end{aligned}
\end{align*}
\end{enumerate}
\end{definition}

We now state our main result on the local existence of a unique weak solution.

\begin{theorem}\label{ma-thm}
Suppose that either Framework A or Framework B hold, and the initial data $f^0$ satisfies
\begin{equation}\label{ini-condi}
f^0 \in (L_+^1 \cap L^p)(\R^d \times \R^d) \cap \pp_1(\R^d \times \R^d).
\end{equation}
Then there exist $T>0$ and a unique weak solution $f$ in the sense of Definition {\rm\ref{def-ext}} on the time interval $[0,T]$. Furthermore, if $f_i,i=1,2$ are two such solutions to \eqref{k-CS-si}, then we have the following $d_1$-stability estimate.
    \[
    \frac{d}{dt}d_1(f_1(t),f_2(t)) \leq Cd_1(f_1(t),f_2(t)), \quad \mbox{for} \quad t \in [0,T],
    \]
where $C$ is a positive constant.
\end{theorem}

%
%
%
%
\section{Local well-posedness of the generalized Cucker-Smale model}\label{sec-local-well}
In this section, we provide a detailed proof of Theorem \ref{ma-thm} in the Framework $A$. Since the arguments for the Framework $B$ are similar to this, we will give a sketch proof for it in the last part of this section. We also notice that it is enough to show Theorem \ref{ma-thm} in the Framework $A$ when $\beta_1=2$ due to the estimate of compact support of $f$ in velocity (see Lemma \ref{lem-gro-vel}). 
%
%
\subsection{A regularized model}
In this part, we will consider a regularized model. 
For this, we first introduce a standard mollifier $\theta$:
\[
\theta(x) = \theta(-x) \geq 0, \quad \theta \in \mathcal{C}_0^{\infty}(\R^d), \quad \mbox{supp }\theta \subset B(0,1), \quad \int_{\R^d} \theta(x) dx = 1,
\]
and we set a sequence of smooth mollifiers:
\[
\theta_{\e}(x) := \frac{1}{\e^d}\theta \left( \frac{x}{\e}\right).
\]
Here $B(0,1):=\{x \in \R^d : |x| \leq 1\}$. Then we regularize $\psi^1$ such as $\psi^1_\e := \psi^1 * \theta_\e$. Since $\psi^1_\e \in \mathcal{C}^{\infty}(\R^d)$, we deduce from well-posedness theories in \cite{CCR,HL,HT} that there exists a unique global solution $f_\e$ which has compact support in kinetic velocity to the following equations.
\begin{equation}\label{reg-k-CS-si}
\left\{ \begin{array}{ll}
\partial_t f_\e + v \cdot \nabla_x f_\e
 + \nabla_v \cdot \big[F_1(f_\e)f_\e\big] = 0, & \qquad (x, v) \in \bbr^d \times \bbr^d,~~t > 0,\\[2mm]
\displaystyle F(f_\e)(x,v,t) := \int_{\R^d \times \R^d} \psi^1_\e(x-y)(w - v)f_\e(y,w,t)dydw, & \qquad (x, v) \in \bbr^d \times \bbr^d,~~t > 0,\\[4mm]
f_\e(x,v,0) =: f^0(x,v), & \qquad (x,v) \in \R^d \times \R^d.
\end{array} \right.
\end{equation}
For the solution $f_\e$ to the system \eqref{reg-k-CS-si}, we will show the uniform $L^p$-bound of $f_\e$ in $\e$. For this, we first need to estimate the growth of the kinetic velocity. Consider the forward bi-characteristics $Z_\e(s) := \left( X_\e(s;0,x,v), V_\e(s;0,x,v)\right)$ satisfying the following ODE system:
\begin{align}\label{traj-vel-est}
\begin{aligned}
\frac{dX_\e(s)}{ds} &= V_\e(s),\cr
\frac{dV_\e(s)}{ds} &= \int_{\R^d \times \R^d} \psi^1_\e\left( X_\e(s) - y\right)(w - V_\e(s))f_\e(y,w,s) dy dw.
\end{aligned}
\end{align}
Set $\Omega_\e(t)$ and $R_\e^v(t)$ the $v$-projection of compact supp$f(\cdot,t)$ and maximum value of $v$ in $\Omega_\e(t)$, respectively:
\begin{equation*}
\Omega_\e(t) := \overline{\{ v \in \R^d : \exists (x,v) \in \R^d \times \R^d \mbox{ such that } f_\e(x,v,t) \neq 0 \}}, \quad R_\e^v(t):= \max_{v \in \Omega_\e(t)}|v|.
\end{equation*}
Then we have the following growth estimate for support of $f_\e$ in velocity.
\begin{lemma}\label{lem-gro-vel} Let $Z_\e(t)$ be the solution to the particle trajectory \eqref{traj-vel-est} issued from the compact supp$_{(x,v)}f^0$ at time $0$. Then we have
\[
R_\e^v(t) \leq R_\e^v(0) = R_0^v:= \max_{v\in \Omega(0)}|v|,
\]
i.e., the support of $f(x,v,t)$ in velocity is uniformly bounded by the one of $f^0(x,v)$.
\end{lemma}
\begin{proof} For the proof, we employ the same idea in \cite[Section 4]{CFRT}. We choose $V_\e(t)$ that make the value of $R_\e^v(t)$ such that $\frac{dR_\e^v(t)}{dt}$ is well-defined to obtain
\begin{align*}
\begin{aligned}
\frac{1}{2}\frac{d}{dt}\left( R_\e^v(t) \right)^2 &= \frac{1}{2}\frac{d}{dt}\left|V_\e(t)\right|^2 = V_\e(t) \cdot \frac{d}{dt}V_\e(t) \cr
&= \int_{\R^d \times \R^d} \psi^1_\e\left(X_\e(t) - y\right)\left(w - V_\e(t)\right)\cdot V_\e(t)f_\e(y,w,t)dydw \cr
&\leq 0.
\end{aligned}
\end{align*}
Here we used the fact that for any $w \in \Omega_\e(t)$, $\left( w - V_\e(t) \right) \cdot V_\e(t) \leq 0.$
This completes the proof.
\end{proof}
\begin{remark}Set $\tilde{\Omega}_0 := B(0,R_0^v)$. Then it follows from Lemma \ref{lem-gro-vel} that $\Omega_\e(t) \subset \tilde{\Omega}_0$ for $t \geq 0$.
\end{remark}
We now show the $L^p$-estimate of $f_\e$ with the help of the estimate in Lemma \ref{lem-gro-vel}.
\begin{proposition}\label{prop-lp} Let $f_\e$ be the solution to the system \eqref{reg-k-CS-si}. Then there exists $T>0$ we have the uniform $L^1\cap L^p$-estimate of $f_\e$:
\[
\sup_{t \in [0,T]}\|f_\e\|_{L^1 \cap L^p} \leq C,
\]
where $C$ is a positive constant independent of $\e$.
\end{proposition}
\begin{proof} We first notice that the conservation of mass to the system \eqref{reg-k-CS-si}:
\[
\frac{d}{dt}\int_{\R^d \times \R^d} f_\e dx dv = 0.
\]
We next turn to $L^p$-estimate of $f_\e$. It is a straightforward to get
\[
\frac{d}{dt}\int_{\R^d \times \R^d} f_\e^p dx dv = -(p-1)\int_{\R^d \times \R^d} \left( \nabla_v \cdot (F_1(f_\e))\right) f_\e^p dx dv.
\]
For the estimate of $\|\nabla_v \cdot (F_1(f_\e))\|_{L^\infty}$, we use a cut-off function $\chi_1 \in \mathcal{C}_c^{\infty}(\R^d)$ defined by
\begin{equation*}
\chi_1(x) := \left\{
\begin{array}{ll}
1 & |x| \leq 1,\\[1mm]
0 & |x| >1.
\end{array} \right.
\end{equation*}
Then it follows from the assumption on the exponent $\alpha$  that
\[
\psi^1_\e(x) = \psi^1 * \theta_\e = (\psi^1(\chi_1 + (1 - \chi_1)))*\theta_\e = (\psi^1\chi_1)*\theta_\e + (\psi^1(1 - \chi_1))*\theta_\e,
\]
and
\[
\|(\psi^1 \chi_1)*\theta_\e\|_{L^{p'}} \leq \|\psi^1 \chi_1\|_{L^{p'}} \leq C, \quad \|(\psi^1(1 - \chi_1))*\theta_\e\|_{L^\infty} \leq \|\psi^1(1 - \chi_1)\|_{L^\infty} \leq 1.
\]
Thus we obtain
\begin{align*}
\begin{aligned}
|\nabla_v \cdot (F_1(f_\e))| &\leq d\int_{\R^d \times \R^d} |(\psi^1 \chi_1) * \theta_\e||f_\e| dy dw + d\int_{\R^d \times \R^d} |(\psi^1(1 - \chi_1))*\theta_\e||f_\e|dy dw \cr
&\leq C(R_0^v)^{\frac{1}{p'}}\|\psi^1 \chi_1\|_{L^{p'}}\|f_\e\|_{L^p} + \|\psi^1(1 - \chi_1)\|_{L^\infty}\|f_\e\|_{L^1}\cr
&\leq C\|f_\e\|_{L^1 \cap L^p},
\end{aligned}
\end{align*}
where $C$ is a positive constant independent of $\e$. Hence we have
\[
\frac{d}{dt}\|f_\e\|_{L^1 \cap L^p} \leq Cd\left(1 - \frac{1}{p}\right) \|f_\e\|^2_{L^1 \cap L^p},
\]
and this yields that there exists $T > 0$,
\[
\sup_{t \in [0,T]}\|f_\e\|_{L^1 \cap L^p} \leq C,
\]
where $C$ is a positive constant independent of $\e$.
\end{proof}

\begin{remark}\label{rmk-1-mom} 1. It is easy to find the estimate of first moments of $f_\e$. In fact, it directly follows from \eqref{reg-k-CS-si} that
\begin{align*}
\begin{aligned}
\frac{d}{dt}\int_{\R^d \times \R^d} |v|f_\e dx dv &= \int_{\R^d \times \R^d} \frac{v}{|v|} \cdot F_1(f_\e)f_\e dx dv \cr
& \leq \int_{\R^{2d} \times \R^{2d}} \psi^1_\e(x-y)|w|f_\e(x,v)f_\e
(y,w) dx dv dy dw \cr
&\qquad- \int_{\R^{2d} \times \R^{2d}} \psi^1_\e(x-y)|v|f_\e(x,v)f_\e
(y,w) dx dv dy dw \cr
&= 0,
\end{aligned}
\end{align*}
where we used $\psi^1_\e(x) = \psi^1_\e(-x)$, and interchange of variables $(x,v) \leftrightarrow (y,w)$. This yields
\begin{equation}\label{fir-mom-est}
\| vf_\e \|_{L^{\infty}(0,T;L^1)} \leq \| vf^0 \|_{L^1}.
\end{equation}
Since
\[
\frac{d}{dt}\int_{\R^d \times \R^d}|x| f_\e dx dv \leq \int_{\R^d \times \R^d}|v| f_\e dx dv,
\]
we deduce from \eqref{fir-mom-est} that
\[
\| xf_\e \|_{L^{\infty}(0,T;L^1)} \leq \| xf^0 \|_{L^1} + T\| vf^0 \|_{L^1}.
\]
2. It follows from the definition of $\psi^1_\e$ that
\begin{align}\label{j-1-est-1}
\begin{aligned}
\psi^1_\e(x) &= \int_{\R^d} \frac{1}{|x-y|^{\alpha}}\theta_\e(y)dy\cr
&\leq \int_{\left\{y: |y| < \frac{|x|}{2}\right\}}
\frac{\theta_{\e}(y)}{|x-y|^{\alpha}}  dy +
\int_{\left\{y: |y| \geq \frac{|x|}{2}\right\}} \frac{
\theta_{\e}(y)}{|x-y|^{\alpha}}  dy \cr &\leq
\frac{2^{\alpha}\e}{|x|^{\alpha}}\int_{\R^d}\theta_{\e}(y)dy
+ \mathbf{1}_{\{|x| \leq 2\e\}}\int_{\left\{y:~ \e \geq |y|
\right\}}\frac{\theta_{\e}(y)}{|x-y|^{\alpha}}dy \cr &\leq
\frac{C}{|x|^{\alpha}} +
\frac{C\e^{\alpha}}{|x|^{\alpha}}\int_{\left\{ y:~ \e \geq
|y| \right\}}\frac{\theta_{\e}(y)}{|x-y|^{\alpha}} dy \leq
\frac{C}{|x|^{\alpha}}.
\end{aligned}
\end{align}
Thus we obtain
\begin{equation}\label{key-ineq}
|\psi^1_\e(x) - \psi^1_\e(y)| \leq \frac{C|x-y|}{\min(|x|,|y|)^{1+\alpha}},
\end{equation}
where $C$ is independent of $\e$.
\end{remark}
We now show the growth estimate of $d_1(f_\e(t),f_{\e'}(t))$.
\begin{proposition}\label{prop-grow} Let $f_\e$ and $f_{e'}$ be two solutions to the system \eqref{reg-k-CS-si}. Then we have
\[
\frac{d}{dt}d_1(f_\e(t),f_{\e'}(t)) \leq C(d_1(f_\e(t),f_{\e'}(t)) + \e + \e'),
\]
for $\e, \e' > 0$. Here $C$ is independent of $\e$ and $\e'$.
\end{proposition}
\begin{proof}We first define flows $Z_\e := (X_\e,V_\e), Z_{\e'}:= (X_{\e'},V_{\e'}) : \R_+ \times \R_+ \times \R^d \times \R^d \to \R^d \times \R^d$ generated from \eqref{reg-k-CS-si} satisfying
\begin{equation}\label{traj}
\left\{ \begin{array}{ll}
\displaystyle \frac{d}{dt} X_\e(t;s,x,v) = V_\e(t;s,x,v), &\\[4mm]
\displaystyle \frac{d}{dt} V_\e(t;s,x,v) = F_1(f_\e)(Z_\e(t;s,x,v),t), &
\\[4mm]
(X_\e(s;s,x,v),V_\e(s;s,x,v)) = (x,v), &
\end{array} \right.
\end{equation}
and
\begin{equation}\label{traj2}
\left\{ \begin{array}{ll}
\displaystyle \frac{d}{dt}X_{\e'}(t;s,x,v) = V_{\e'}(t;s,x,v), &\\[4mm]
\displaystyle \frac{d}{dt} V_{\e'}(t;s,x,v) = F_1(f_{\e'})(Z_{\e'}(t;s,x,v),t), &
\\[4mm]
(X_{\e'}(s;s,x,v),V_{\e'}(s;s,x,v)) = (x,v), &
\end{array} \right.
\end{equation}
for all $s,t \in [0,T]$. Since $\psi^1_{\e}, \psi^1_{\e'} \in \mathcal{C}^{\infty}$, \eqref{traj} and \eqref{traj2} are well-defined for $s,t \in [0,T]$. We now choose an optimal transport map $\mt^0 = (\mt_1^0(x), \mt_2^0(v))$ between $f_\e(t_0)$ and $f_{\e'}(t_0)$ for fixed $t_0 \in [0,T)$, i.e., $f_{\e'}(t_0) = \mt^0 \# f_\e(t_0)$. It is known from \cite{CPJ} that such an optimal transport map exists when $f_\e(t_0)$ is absolutely continuous with respect to the Lebesgue measure. Then we apply the similar argument in \cite[Lemma 5.5]{HL} to obtain $f_\e(t) = Z_\e(t;t_0,\cdot,\cdot)\#f_\e(t_0)$ and $f_{\e'}(t) = Z_\e(t;t_0,\cdot,\cdot)\#f_{\e'}(t_0)$ using the mass transportation (not necessarily optimal) notation of push-forward. More precisely, we obtain that for any $g \in \mathcal{C}_c^1(\R^d \times \R^d \times [0,T])$,
\begin{align}\label{push-f}
\begin{aligned}
&\int_{\R^d \times \R^d} g(x,v,t)f_\e(x,v,t) dx dv - \int_{\R^d \times \R^d} g(x,v,t_0)f_\e(x,v,t_0) dx dv \cr
& \qquad = \int_{t_0}^t \int_{\R^d \times \R^d} \left( \partial_s g + v\cdot\nabla_x g + F_1(f_\e) \cdot \nabla_v g\right)f_\e(x,v,s)dxdvds.
\end{aligned}
\end{align}
We now choose
\[
g(x,v,t) := h(X_\e(s;t,x,v),V_\e(s;t,x,v)), \quad \mbox{for fixed } t,
\]
where $h \in \mathcal{C}_c^1(\R^d \times \R^d)$. This makes the r.h.s of \eqref{push-f} vanished and
\begin{equation}\label{push-f-2}
\int_{\R^d \times \R^d} h(x,v)f_\e(x,v,t) dx dv = \int_{\R^d \times \R^d}h(X_\e(t_0;t,x,v),V_\e(t_0;t,x,v))f_\e(x,v,t_0)dxdv.
\end{equation}
Thus we conclude $f_\e(t) = Z_\e(t;t_0,\cdot,\cdot) \#f_\e(t_0)$. Same argument can be applied to get $f_{\e'}(t) = Z_\e(t;t_0,\cdot,\cdot)\#f_{\e'}(t_0)$. We also notice that
\[
\mt^t \# f_\e(t) = f_{\e'}(t), \quad \mbox{where} \quad \mt^t = Z_{\e'}(t;t_0,\cdot,\cdot) \circ \mt^0 \circ Z_\e(t_0;t,\cdot,\cdot).
\]
By Definition \ref{defdp}, when $p=1$, we obtain
\[
d_1(f_\e(t),f_{\e'}(t)) \leq \int_{\R^d \times \R^d} |Z_\e(t;t_0,x,v) - Z_{\e'}(t;t_0,\mt^0(x,v))|f_\e(x,v,t_0) dx dv.
\]
Set
\[
Q_{\e,\e'}(t) := \int_{\R^d \times \R^d} |Z_\e(t;t_0,x,v) - Z_{\e'}(t;t_0,\mt^0(x,v))|f_\e(x,v,t_0) dx dv.
\]
Then straightforward computations yield
\begin{align*}
\begin{aligned}
&\frac{d}{dt}Q_{\e,\e'}(t) \Big|_{t = t_0+} \cr
&\leq \int_{\R^d \times \R^d} |V_\e(t;t_0,x,v) - V_{\e'}(t;t_0,\mt^0(x,v))|f_\e(x,v,t_0)dxdv \Big|_{t = t_0+} \cr
& + \int_{\R^d \times \R^d} \left| F_1(f_\e)(Z_\e(t;t_0,x,v),t) -F_1(f_{\e'})(Z_{\e'}(t;t_0,\mt^0(x,v)),t) \right|f_\e(x,v,t_0)dx dv \bigg|_{t = t_0+} \cr
&=: \mathcal{I} + \mathcal{J}.
\end{aligned}
\end{align*}
For the estimate of $\mathcal{I}$, it is easy to find
\begin{equation}\label{i-est}
\mathcal{I} = \int_{\R^d \times \R^d} |v - \mt^0_2(v)|f_\e(x,v,t_0) dx dv \leq C d_1(f_\e(t_0),f_{\e'}(t_0)).
\end{equation}
For the estimate of $\mathcal{J}$, we notice that
\begin{align*}
\begin{aligned}
\mathcal{J} &= \int_{\R^d \times \R^d} \bigg| \int_{\R^d \times \R^d} \psi^1_\e(x - y)(w - v)f_\e(y,w,t_0)dy dw \cr
&\quad -\int_{\R^d \times \R^d} \psi^1_{\e'}(\mt^0_1(x)- y)(w - \mt^0_2(v))f_{\e'}
(y,w,t_0)dy dw \bigg|f_\e(x,v,t_0)dx dv \cr
&= \int_{\R^d \times \R^d} \bigg| \int_{\R^d \times \R^d} \psi^1_\e(x - y)(w - v)f_\e(y,w,t_0)dy dw \cr
&\quad -\int_{\R^d \times \R^d} \psi^1_{\e'}(\mt^0_1(x)- \mt^0_1(y))(\mt^0_2(w) - \mt^0_2(v))f_{\e}
(y,w,t_0)dy dw \bigg|f_\e(x,v,t_0)dx dv.
\end{aligned}
\end{align*}
For notational simplicity, we omit the time dependency on $t_0$ in the rest of computations. We decompose $\mathcal{J}$ into two parts: 
\[
\mathcal{J} = \int_{\R^d \times \R^d} | \mathcal{J}_1 + \mathcal{J}_2 | f_\e(x,v) dx dv,
\]
where
\begin{align*}
\begin{aligned}
\mathcal{J}_1 &:= \int_{\R^d \times \R^d} \left(\psi^1_\e(x-y) - \psi^1_{\e'}\left(\mt^0_1(x)-\mt^0_1(y)\right)\right)(w-v)f_\e(y,w)dydw,\cr
\mathcal{J}_2 &:= \int_{\R^d \times \R^d} \psi^1_{\e'}\left(\mt^0_1(x) - \mt^0_1(y)\right)\left( (w-v) - \left(\mt^0_2(w) - \mt^0_2(v)\right)\right)f_\e(y,w)dydw.
\end{aligned}
\end{align*}
For the estimates of $\mathcal{J}$, we divide it into two steps to make the reading easier.

\begin{itemize}
\item In Step A, we show 
\begin{equation}\label{claim}
\int_{\R^d \times \R^d}|\mathcal{J}_1|f_\e dxdv \leq C\max(\|f_\e\|, \|f_{\e'}\|) d_1(f_\e(t_0), f_{\e'}(t_0)) + C \|f_\e\|^2(\e + \e'),
\end{equation}
where $C$ is independent of $\e$. \newline
\item In Step B, we show
\[
\int_{\R^d \times \R^d}|\mathcal{J}_2|f_\e dxdv \leq C\|f_{\e}\|\|f_{\e'}\|d_1(f_\e(t_0),f_{\e'}(t_0)).
\]
\end{itemize}

\textbf{Step A}: By adding and subtracting, we find that
\begin{align}\label{j-1-est-2}
\begin{aligned}
\mathcal{J}_1 &\leq \int_{\R^d \times \R^d} \left|(\psi^1_\e - \psi^1_{\e'})(x-y)\right||w-v|f_\e(y,w)dydw \cr
&\quad + \int_{\R^d \times \R^d} \left|\psi^1_{\e'}(x-y) - \psi^1_{\e'}\left(\mt^0_1(x)-\mt^0_1(y)\right)\right||w-v|f_\e(y,w)dydw.
\end{aligned}
\end{align}
It follows from a similar estimate to \eqref{j-1-est-1} that
\begin{align}\label{j-1-est-3}
\begin{aligned}
|\psi^1_\e(x) - \psi^1(x)| &\leq \int_{\R^d} |\psi^1(x-y) - \psi^1(x)|\theta_\e(y) dt\cr
&\leq 2\int_{\R^d} \left( \frac{1}{|x|^{1+\alpha}} + \frac{1}{|x-y|^{1+\alpha}}\right) |y|\theta_\e(y)dy \cr
&\leq 2\e\int_{\left\{ y:~ \e \geq
|y| \right\}} \left( \frac{1}{|x|^{1+\alpha}} + \frac{1}{|x-y|^{1+\alpha}}\right)\theta_\e(y)dy \cr
&\leq \frac{C\e}{|x|^{1+\alpha}}.
\end{aligned}
\end{align}
Then we use \eqref{j-1-est-3} to obtain
\begin{align}
&\int_{\R^{2d} \times \R^{2d}} \left|(\psi^1_\e - \psi^1)(x-y)\right||w-v|f_\e(y,w)f_\e(x,v)dxdvdydw \cr
& \quad \leq 2R_0^v\int_{\R^{2d} \times \tilde{\Omega}_0^2} \left|(\psi^1_\e - \psi^1)(x-y)\right|f_\e(y,w)f_\e(x,v)dxdvdydw\cr
& \quad \leq C\e\int_{\R^{2d} \times \tilde{\Omega}_0^2} \frac{1}{|x-y|^{1+\alpha}}f_\e(y,w)f_\e(x,v)dxdvdydw \cr
&\quad \leq C\e\int_{\R^d \times \R^d} \left( \int_{\{ y : |x-y| < 1\} \times \tilde{\Omega}_0} + \int_{\{ y : |x-y| \geq 1\} \times \tilde{\Omega}_0} \frac{1}{|x-y|^{1+\alpha}}f_\e(y,w) dydw\right)f_\e(x,v) dxdv\cr
&\quad \leq C\e\int_{\R^d \times \R^d} \left(\left(\int_{\{ y : |x-y|\leq 1\}}\frac{1}{|x-y|^{(1+\alpha)p'}}dy\right)^{\frac{1}{p'}}\|f_\e\|_{L^p} +\|f_\e\|_{L^1}\right)f_\e(x,v) dxdv\cr
&\quad\leq C\e\|f_\e\|^2 \leq C\e.\label{j-1-est-4}
\end{align}
Similarly, we get
\[
\int_{\R^{2d} \times \R^{2d}} \left|(\psi^1_{\e'} - \psi^1)(x-y)\right||w-v|f_\e(y,w)f_\e(x,v)dxdvdydw \leq C\e'.
\]
Thus we have
\begin{equation}\label{j-1-est-5}
\int_{\R^{2d} \times \R^{2d}} \left|(\psi^1_\e - \psi^1_{\e'})(x-y)\right||w-v|f_\e(y,w)f_\e(x,v)dxdvdydw \leq C(\e+\e').
\end{equation}
For a related to second term in the rhs of \eqref{j-1-est-2}, we employ \eqref{key-ineq} and interchange the variables $(x,v) \leftrightarrow (y,w)$ to find
\begin{align*}
\begin{aligned}
&\int_{\R^{2d} \times \tilde{\Omega}_0^2} \left|\psi^1_{\e'}(x-y) - \psi^1_{\e'}\left(\mt^0_1(x)-\mt^0_1(y)\right)\right||w-v|f_\e(x,v)f_\e(y,w)dxdvdydw\cr
&\quad \leq 2R_0^v\int_{\R^{2d} \times \tilde{\Omega}_0^2} \left( \frac{|\mt^0_1(x) - x|}{|\mt^0_1(x) - \mt^0_1(y)|^{1+\alpha}} + \frac{|\mt^0_1(x) - x|}{|x-y|^{1+\alpha}} \right)f_\e(x,v)f_\e(y,w)dxdvdydw \cr
&=:\mathcal{K}_1 + \mathcal{K}_2.
\end{aligned}
\end{align*}
By direct computations, we get
\begin{align*}
\begin{aligned}
\mathcal{K}_1 &= 2R_0^v \int_{\R^d \times \R^d} |\mt^0_1(x) - x|f_\e(x,v) \left( \int_{\R^d \times \Omega_0} \frac{1}{|\mt^0_1(x) - y|^{1+\alpha}}f_{\e'}(y,w) dydw\right) dxdv \cr
&\leq C\|f_{\e'}\|\int_{\R^d \times \R^d} |\mt^0_1(x) - x|f_\e(x,v) dxdv \leq C\|f_{\e'}\|d_1(f_\e(t_0),f_{\e'}(t_0)),
\end{aligned}
\end{align*}
where we used the same estimates in \eqref{j-1-est-4} to obtain
\[
\int_{\R^d \times \tilde{\Omega}_0} \frac{1}{|\mt^0_1(x) - y|^{1+\alpha}}f_{\e'}(y,w) dydw \leq C\|f_{\e'}\|.
\]
Similarly, we also obtain $\mathcal{K}_2 \leq C\|f_{\e}\|d_1(f_\e(t_0),f_{\e'}(t_0))$. This yields
\begin{align}\label{j-1-est-6}
\begin{aligned}
&\int_{\R^{2d} \times \R^{2d}} \left|\psi^1_{\e'}(x-y) - \psi^1_{\e'}\left(\mt^0_1(x)-\mt^0_1(y)\right)\right||w-v|f_\e(x,v)f_\e(y,w)dxdvdydw\cr
& \quad \qquad \leq C\max(\|f_\e\|, \|f_{\e'}\|) d_1(f_\e(t_0),f_{\e'}(t_0)).
\end{aligned}
\end{align}
We now combine \eqref{j-1-est-5} and \eqref{j-1-est-6} to conclude our desired claim. \newline

\textbf{Step B}: For the estimate of $\mathcal{J}_2$, we obtain
\begin{align*}
\begin{aligned}
\mathcal{J}_2 &\leq \int_{\R^{2d} \times \R^{2d}} \left|\psi^1_{\e'}\left( \mt^0_1(x) - \mt^0_1(y)\right)\right| \left| w - \mt^0_2(w)\right|f_\e(x,v) f_\e(y,w) dx dy dv dw \cr
&  \quad + \int_{\R^{2d} \times \R^{2d}} \left|\psi^1_{\e'}\left( \mt^0_1(x) - \mt^0_1(y)\right)\right| \left| v - \mt^0_2(v)\right| f_\e(x,v) f_\e(y,w) dx dy dv dw \cr
&=: \mathcal{J}_2^1 + \mathcal{J}_2^2.
\end{aligned}
\end{align*}
On the other hand, we can easily find that
\begin{align*}
\begin{aligned}
\mathcal{J}_2^1 &= \int_{\R^d \times \R^d} \left( \int_{\R^d \times \Omega_0 }\left|\psi^1_{\e'}\left( \mt^0_1(x) - \mt^0_1(y)\right)\right| f_\e(x,v) dx dv \right) \left|w - \mt^0_2(w)\right| f_\e(y,w)dydw \cr
&\leq C\|f_{\e'}\|\int_{\R^d \times \R^d}|w - \mt^0_2(w)|f_\e(y,w)dydw \leq C\|f_{\e'}\|d_1(f_\e(t_0),f_{\e'}(t_0)),
\end{aligned}
\end{align*}
where we used the estimates in \eqref{j-1-est-4} again. Similarly, we get
\[
\mathcal{J}_2^2 \leq C\|f_{\e'}\|d_1(f_\e(t_0),f_{\e'}(t_0)),
\]
and this deduces
\begin{equation}\label{j-2-est-1}
\mathcal{J}_2 \leq C\|f_{\e'}\|d_1(f_\e(t_0),f_{\e'}(t_0)).
\end{equation}
Thus we have
\[
\int_{\R^d \times \R^d}|\mathcal{J}_2|f_\e dxdv \leq C\|f_{\e}\|\|f_{\e'}\|d_1(f_\e(t_0),f_{\e'}(t_0)).
\]
We now combine \eqref{i-est}, \eqref{claim} and \eqref{j-2-est-1} to find
\[
\frac{d}{dt}Q_{\e,\e'}(t) \Big|_{t = t_0+} \leq  C(d_1(f_\e(t_0), f_{\e'}(t_0)) + \e + \e').
\]
We finally write the integral form, dividing $t - t_0$, and taking the limit $t \to t_0^+$ to conclude
\[
\frac{d}{dt}d_1(f_\e(t),f_{\e'}(t))\Big|_{t = t_0^+} \leq C(d_1(f_\e(t_0), f_{\e'}(t_0) + \e + \e').
\]
Since $t_0$ is arbitrary in $[0,T)$, this yields
\begin{equation*}
\frac{d}{dt}d_1(f_\e(t),f_{\e'}(t)) \leq C(d_1(f_\e(t_0), f_{\e'}(t_0) + \e + \e'),
\end{equation*}
where $C$ is independent of $\e$ and $\e'$.
\end{proof}
\subsection{Existence and uniqueness of weak solutions (limit as $\e \to 0$)}\label{sec3.2} It follows from Proposition \ref{prop-grow} that $\{f_\e\}_{\e >0}$ is a Cauchy sequence in $\mathcal{C}([0,T];\pp_1(\R^d \times \R^d))$, and this implies that there exists a limit curve of measure $f \in \mathcal{C}([0,T];\pp_1(\R^d \times \R^d))$, and $f \in L^{\infty}(0,T;(L_+^1 \cap L^p)(\R^d \times \R^d))$. Thus it only remains to show that $f$ is a solution of the Cucker-Smale model \eqref{k-CS-si}. Choose a test function $\Psi(x,v,t) \in \mathcal{C}_c^{\infty}(\R^d \times \R^d \times [0,T])$, then $f_\e$ satisfies
\begin{align}\label{weak-var}
\begin{aligned}
&\int_{\R^d \times \R^d} \Psi^0(x,v)f^0(x,v) dx dv \cr
& \quad = \int_{\R^d \times \R^d} \Psi(x,v,T)f_\e(x,v,T)dx dv + \int_0^T\int_{\R^d \times \R^d} f_\e(x,v,t)\partial_t\Psi(x,v,t)dxdvdt \cr
& \qquad -\int_0^T\int_{\R^d \times \R^d} (\nabla_x \Psi) \cdot v f_\e dx dvdt - \int_0^T\int_{\R^d \times \R^d} (\nabla_v \Psi) \cdot F_1(f_\e)f_\e dx dv dt.
\end{aligned}
\end{align}
We can easily show that the first, second, and third terms in the rhs of \eqref{weak-var} converge to
\begin{align*}
\begin{aligned}
&\int_{\R^d \times \R^d} \Psi(x,v,T)f(x,v,T)dx dv + \int_0^T\int_{\R^d \times \R^d} f(x,v,t)\partial_t\Psi(x,v,t)dxdvdt \cr
& \qquad -\int_0^T\int_{\R^d \times \R^d} (\nabla_x \Psi) \cdot v f dx dvdt \quad \mbox{as} \quad \e \to 0,
\end{aligned}
\end{align*}
since $f_\e \to f$ in $\mathcal{C}([0,T],\pp_1(\R^d \times \R^d))$. We also notice that
\begin{align*}
\begin{aligned}
&\left| \int_0^T \int_{\R^{2d} \times \R^{2d}} (\psi^1_\e - \psi^1)(x-y)(\nabla_v \Psi) \cdot (w-v)f_\e(x,v)f_\e(y,w) dxdvdydwdt \right| \cr
&\leq C\e R_0^v\|\nabla_v \Psi\|_{L^{\infty}(0,T;L^{\infty}(\R^d \times \R^d))}\int_0^T\int_{\R^{2d} \times \tilde{\Omega}_0^2} \frac{1}{|x-y|^{1+\alpha}}f_\e(x,v)f_\e(y,w) dxdvdydwdt \cr
& \leq C\e\|f_\e\|^2 \leq C\e \to 0 \quad \mbox{as} \quad \e \to 0\,,
\end{aligned}
\end{align*}
where we used the decomposition in local and far fields as in \eqref{j-1-est-4}. 
Thus in order to obtain
\[
\int_0^T\int_{\R^d \times \R^d} (\nabla_v \Psi) \cdot F_1(f_\e)f_\e dx dv dt \to \int_0^T\int_{\R^d \times \R^d} (\nabla_v \Psi) \cdot F_1(f)f dx dv dt,
\]
it only remains to show
\begin{align}\label{ext-est-1}
\begin{aligned}
&\int_0^T \int_{\R^{2d} \times \R^{2d}} \psi^1(x-y) (\nabla_v \Psi) \cdot (w-v)f_\e(x,v)f_\e(y,w) dxdvdydwdt \cr
& \qquad \quad \to \int_0^T \int_{\R^{2d} \times \R^{2d}} \psi^1(x-y) (\nabla_v \Psi) \cdot (w-v)f(x,v)f(y,w) dxdvdydwdt,
\end{aligned}
\end{align}
as $\e \to 0$. For this, we introduce a cut-off function
$\chi_{\delta} \in \mathcal{C}_c^{\infty}(\R^d)$ such that
\[
\chi_{\delta}(x) = \left\{  \begin{array}{ll}
1 & \mbox{if} \quad |x|\leq \delta \\
0 & \mbox{if} \quad |x|\geq 2\delta
\end{array} \right.\,.
\]
Then since $(1 - \chi_{\delta}(x-y))\psi^1(x-y)(w-v)\cdot \nabla_v \Psi$ is a Lipschitz function and $f_\e \to f$ in $\mathcal{C}([0,T],\pp_1(\R^d \times \R^d))$, we find
\begin{align}\label{ext-est-2}
\begin{aligned}
&\int_0^T \int_{\R^{2d} \times \R^{2d}} (1 - \chi_{\delta})\psi^1(x-y) (\nabla_v \Psi) \cdot (w-v)f_\e(x,v)f_\e(y,w) dxdvdydwdt \cr
& \qquad \quad \to \int_0^T \int_{\R^{2d} \times \R^{2d}}(1 - \chi_{\delta}) \psi^1(x-y) (\nabla_v \Psi) \cdot (w-v)f(x,v)f(y,w) dxdvdydwdt,
\end{aligned}
\end{align}
as $\e \to 0$ for any $\delta > 0$.
On the other hand, the remaining term is estimated as follows:
\begin{align}\label{ext-est-3}
\begin{aligned}
&\int_0^T \int_{\R^{2d} \times \R^{2d}} \chi_{\delta}(x-y)\psi^1(x-y) (\nabla_v \Psi) \cdot (w-v)f_\e(x,v)f_\e(y,w) dxdvdydwdt \cr
&\quad \leq C\delta\int_0^T\int_{\{(x,y)\in \R^d \times \R^d :~|x-y|\leq
2\delta\} \times \tilde{\Omega}_0^2} \frac{1}{|x-y|^{1+\alpha}}f_\e(x,v)f_\e(y,w) dxdvdydwdt \cr
&\quad \leq C\delta \to 0 \quad \mbox{as} \quad \delta \to 0,
\end{aligned}
\end{align}
and similarly, we also have
\begin{equation}\label{ext-est-4}
\int_0^T \int_{\R^{2d} \times \R^{2d}}\chi_{\delta}(x-y) \psi^1(x-y) (\nabla_v \Psi) \cdot (w-v)f(x,v)f(y,w) dxdvdydwdt \leq C\delta \to 0,
\end{equation}
as $\delta \to 0$ due to the fact that $f$ has a compact support in velocity. Hence we conclude the convergence \eqref{ext-est-1} combining \eqref{ext-est-2}, \eqref{ext-est-3}, and \eqref{ext-est-4}. Uniqueness of the weak solutions $f_\e$ is just followed from Proposition \ref{prop-grow}. More specifically, let $f_1, f_2 \in L^{\infty}(0,T;(L_+^1\cap L^p)(\R^d \times \R^d)) \cap \mathcal{C}([0,T],\pp_1(\R^d \times \R^d))$ be the weak solutions to the system \eqref{k-CS-si} with same initial data $f^0 \in (L_+^1 \cap L^p)(\R^d \times \R^d) \cap \pp_1(\R^d \times \R^d)$ satisfying the framework $A$. Then Proposition \ref{prop-grow} yields that
\[
\frac{d}{dt}d_1(f_1(t),f_2(t)) \leq C\max( \|f_1\|, \|f_2\|)d_1(f_1(t),f_2(t)), \quad \mbox{for} \quad t \in [0,T].
\]
This completes the proof in the framework A.
\begin{proposition} Let $f$ be a weak solutions to \eqref{k-CS-si} on the time-interval $[0,T)$ in the sense of Definition {\rm\ref{def-ext}}. Then $f$ is determined as the push-forward of the initial density through the flow map generated by $(v,F_i(f))$.
\end{proposition}
\begin{proof}
Cconsider the following flow map:
\begin{equation}\label{flow}
\left\{ \begin{array}{ll}
\displaystyle \frac{d}{dt}X(t;s,x,v) = V(t;s,x,v), & \\[2mm]
\displaystyle \frac{d}{dt}V(t;s,x,v) = F_1(f)(X(t;s,x,v),V(t;s,x,v),t), &\\[4mm]
(X(s;s,x,v),V(s;s,x,v)) = (x,v),&
\end{array} \right.
\end{equation}
for all $s,t\in[0,T]$. Then since $f$ has compact support in $v$, the flow map \eqref{flow} is well-defined using the same argument in the proof of Proposition \ref{prop-grow}. Moreover we can use the similar argument to \eqref{push-f-2} to have
\[
\int_{\R^d \times \R^d} h(x,v)f(x,v,t) dx dv = \int_{\R^d \times \R^d}h(X(0;t,x,v),V(0;t,x,v))f^0(x,v)dxdv,
\]
for $t \in [0,T]$. This yields that $f$ is determined as the push-forward of the initial density through the flow map \eqref{flow}. 
\end{proof}
\begin{proof}[Proof of Theorem \ref{ma-thm} in the framework $B$] Similarly, we first regularize the nonlinear velocity coupling as $\nabla \phi_\e(v) := \frac{v}{|v|^{2-\beta_2} + \e}$, and define the $f_\e$ by this regularized system. Then we easily find that the estimates of support of $f$ in position and velocity, and first momentums using the same arguments in Lemma \ref{lem-gro-vel} and Remark \ref{rmk-1-mom}. The remaining parts are obtained by similar arguments as in Section \ref{sec-local-well}.
\end{proof}
%
%
%
%

\section{Rich dynamics of the generalized Cucker-Smale particle system}
In this part, we investigate the dynamics of the generalized Cucker-Smale particle system.
We consider Cucker-Smale particle system:
\begin{align}\label{par-gene}
\begin{aligned}
\frac{dx_i(t)}{dt} &= v_i(t),\cr
\frac{dv_i(t)}{dt} &= \frac1N \sum_{j=1}^N \psi^k(|x_j(t) - x_i(t)|)\frac{v_j - v_i}{|v_j - v_i|^{2-\beta}}, \quad k = 1,2,
\end{aligned}
\end{align}
subject to $(x_i(0),v_i(0)) =: (x_i^0,v_i^0)$. Here $\psi^k,k=1,2$ are given in \eqref{condi:psi}. We assume further that $\psi^k(s),k=1,2$ satisfy $\alpha \geq 1$ and
\[
0 < \psi^2(s_1) \leq \psi^2(s_2) \quad \mbox{for} \quad 0 \leq s_2 \leq s_1 < \infty, \quad \mbox{and} \quad \psi^2(s) \to 0 \quad \mbox{as} \quad s \to \infty,
\]
respectively. 

For the system \eqref{par-gene} with $\beta=2$, there are only two papers \cite{ACHL,Pes} on the existence theory and large time behaviour to our best knowledge. In \cite{ACHL}, the authors identified the initial configurations that prevent the pairwise collisions in a finite-time when the singularity of communication weight is strong enough such as $\alpha \geq 1$. One-dimensional discrete C-S model \eqref{par-gene} was treated in \cite{Pes}, and showed existence  of piecewise weak solutions when the singularity of communication weight is sufficiently weak, $\alpha < 1$. 

Note that the original alignment force  of Cucker-Smale model satisfies the conditions of $\psi^2$. Without loss of generality, we may assume that
\[
\sum_{j=1}^Nv_i(0) = 0.
\]
We set 
\[
\|x\|_\infty := \max_{1\leq i \leq N}|x_i| \quad \mbox{and} \quad \|v\|_\infty := \max_{1\leq i \leq N}|v_i|.
\]
We notice that we can choose an index $i$ such that $\|v(t)\|_\infty = |v_i(t)|$ at any time $t$. Then it follows from \cite{CCR2,HHK} that
\[
\frac{d}{dt}\|v(t)\|_\infty^2 \leq -C_0\psi^k(2\|x(t)\|_\infty)\|v(t)\|_\infty^\beta, \quad \mbox{for} \quad \beta \in (0,3).
\]
It is also clear to obtain $\left| \frac{d\|x\|_\infty}{dt}\right| \leq \|v\|_\infty$. We now take the similar argument in \cite{ACHL,HL}, and define two Lyanpunov type functionals $\mathcal{E}_\pm(x,v)$:
\[
\mathcal{E}_\pm(x(t),v(t)) := \frac{1}{3-\beta}\|v(t)\|_\infty^{3-\beta} \pm \frac{C_0}{2}\Psi^k(2\|x(t)\|_\infty),
\]
where $\Psi^k(\cdot)$ is a primitive of $\psi^k$.

We next present two lemmas that can be obtained using the similar argument in \cite{ACHL,HL}.
\begin{lemma} Let $(x,v)$ be any smooth solutions to the system \eqref{par-gene}.
Then we have
\begin{align*}
\begin{aligned}
&(i) \,\,\,\, \mathcal{E}_\pm(x(t),v(t)) \leq \mathcal{E}_\pm(x_0,v_0).\cr
&(ii) \,\, \|v(t)\|_\infty^{3-\beta} + \frac{(3-\beta)C_0}{2}\left|\int_{2\|x_0\|_\infty}^{2\|x(t)\|_\infty} \psi^k(s) ds\right| \leq \|v_0\|_\infty^{3-\beta}.
\end{aligned}
\end{align*}
\end{lemma}
\begin{lemma}\label{lem:key2} Let $(x,v)$ be any smooth solutions to the system \eqref{par-gene}. If the initial data $(x_0,v_0)$ satisfies
\begin{equation}\label{basic-ini}
\|x_0\|_\infty > 0, \quad \|v_0\|_\infty^{3-\beta} < \frac{(3-\beta)C_0}{2}\min \left\{ \int_0^{2\|x_0\|_\infty} \psi^k(s) ds, \int_{2\|x_0\|_\infty}^\infty \psi^k(s)ds\right\},
\end{equation}
then there exist positive constants $x_m,x_M > 0$ such that
\[
\|x(t)\|_\infty \in [x_m, x_M], \quad \frac{d}{dt}\|v(t)\|_\infty^2 \leq -C_0\psi^k(2x_M)\|v(t)\|_\infty^\beta,
\]
where $x_m$ and $x_M$ are defined by
\[
\|v_0\|_\infty^{3-\beta} = \frac{(3-\beta)C_0}{2}\int_{2x_m}^{\|2x_0\|_\infty} \psi^k(s) ds \quad \mbox{and} \quad \quad \|v_0\|_\infty^{3-\beta} = \frac{(3-\beta)C_0}{2}\int_{2\|x_0\|_\infty}^{2x_M} \psi^k(s) ds,
\]
respectively.
\end{lemma}
\begin{remark}Note that if $\alpha \in [1,d-1)$ for $d>2$, then $\int_0^{2\|x_0\|_\infty}\psi^1(s) ds = \infty$ and this yields that we only need the following condition for $v_0$ in Lemma \ref{lem:key2}:
\[
\|v_0\|_\infty^{3-\beta} < \frac{(3-\beta)C_0}{2}\int_{2\|x_0\|_\infty}^\infty \psi^1(s) ds.
\]
\end{remark}
\begin{theorem} Let $(x,v)$ be any smooth solutions to the system \eqref{par-gene} with initial data $(x_0,v_0)$ satisfying \eqref{basic-ini}. Then the followings hold:\newline

$\bullet$ If $\beta = 2$, we have an exponential alignment between velocities:
\[
\|v(t)\|_\infty \leq \|v_0\|_\infty \exp\left\{ -\frac{C_0\psi^k(2x_M)t}{2}\right\}.
\]
Furthermore, if $\eta^0_{m,X} > \frac{\|v_0\|_\infty}{C_0\psi^k(2x_M)}$, then we have no finite-time collision between particles and 
\[
\|v(t)\|_\infty \geq \|v_0\|_\infty \exp\left\{ - \psi^k(\eta_{m,X}^* )t\right\},
\]
where $\eta_{m,X}^*:= \eta^0_{m,X} - \frac{\|v_0\|_\infty}{C_0\psi^k(2x_M)} > 0$.\newline

$\bullet$ If $\beta \in (0,2)$, we have a finite-time alignment between velocities:
\[
\|v(t)\|_\infty \leq \left( \|v_0\|_\infty^{2-\beta} - \frac{(2-\beta)C_0\psi^k(2x_M)t}{2}\right)^{\frac{1}{2-\beta}}.
\]
Furthermore if $\eta^0_{m,X} > T^*\|v_0\|_\infty$, then we have no collision between particles, where
\[
T^* := \frac{4\|v_0\|^{2-\beta}}{(2-\beta)C_0\psi^k(2x_M)}.
\]

$\bullet$ If $\beta \in (2,3)$, we have an polynomial alignment between velocities:
\[
\|v(t)\|_\infty \leq \left( \|v_0\|_\infty^{2-\beta} + \frac{(\beta-2)C_0\psi^k(2x_M)t}{2}\right)^{-\frac{1}{\beta-2}}.
\]
\end{theorem}
\begin{proof} The inequalities for $\|v(t)\|_\infty$ are clearly obtained from the results in Lemma \ref{lem:key2}. Concerning the initial configuration for avoiding collisions between particles, a straightforward computation yields that for $\beta = 2$
\begin{align*}
\begin{aligned}
|\eta_{m,X}(t) - \eta_{m,X}^0| &\leq \left| \int_0^t \frac{d\eta_{m,X}(s)}{ds} ds\right|\cr
&\leq 2\|v_0\|_\infty\int_0^t e^{-\frac{C_0\psi^k(2x_M)s}{2}}ds\cr
&\leq \frac{\|v_0\|_\infty}{C_0\psi^k(2x_M)}.
\end{aligned}
\end{align*}
Thus we conclude that
\[
\eta_{m,X}(t) \geq \eta_{m,X}^0 - |\eta_{m,X}(t) - \eta_{m,X}^0| \geq \eta_{m,X}^0 - \frac{\|v_0\|_\infty}{C_0\psi^k(2x_M)} > 0.
\]
Similarly, for $\beta \in (0,2)$, we have
\[
|\eta_{m,X}(t) - \eta_{m,X}^0| \leq \int_0^t \|v(s)\|_\infty ds \leq T^*\|v_0\|_\infty,
\]
and this deduces 
\[
\eta_{m,X}(t) \geq \eta_{m,X}^0 - T^*\|v_0\|_\infty > 0.
\]
This completes the proof.
\end{proof}
\begin{remark} In the case of $\beta = 2$, if we choose the initial data for position $x_0$ such that $\eta^0_{m,X} > \frac{\|v_0\|_\infty}{C_0\psi^k(2x_M)}$, then there is no collision between particles and alignment for velocities in a finite time. Similarly, if we select the initial data $x_0$ satisfying $\eta^0_{m,X} > T^*\|v_0\|_\infty$ when $\beta \in (0,2)$, then the particles do not collide each other until $T^*$.
\end{remark}
%
%
%
%

\section*{Acknowledgments}
JAC was partially supported by the project MTM2011-27739-C04-02
DGI (Spain) and 2009-SGR-345 from AGAUR-Generalitat de Catalunya.
JAC acknowledges support from the Royal Society by a Wolfson
Research Merit Award. YPC was partially supported by Basic Science Research Program through the National Research Foundation of Korea funded
by the Ministry of Education, Science and Technology (ref.
2012R1A6A3A03039496). JAC and YPC were supported by Engineering
and Physical Sciences Research Council grants with references
EP/K008404/1 (individual grant) and EP/I019111/1 (platform grant).

%
%
%
%


\begin{thebibliography}{99}

\bibitem{ACHL} S. Ahn, H. Choi, S.-Y. Ha, and H. Lee, \textit{On the collision avoiding initial-congurations to the Cucker-Smale type 
flocking models}, Comm. Math. Sci., {\bf 10}, (2012), 625--643.


\bibitem{Bir} B. Birnir, \textit{An ODE model of the motion of pelagic fish}, J. Statist. Phys., {\bf 128}, (2007), 535--568.

\bibitem{CDF} S. Camazine, J.-L. Deneubourg, N. R. Franks, J. Sneyd, G. Theraulaz, and E. Bonabeau, \textit{Self-Organization in Biological systems}, Princeton Univ. Press, 2003.

\bibitem{CCR} J. A. Ca\~nizo, J. A. Carrillo, and J. Rosado, \textit{A well-posedness thoery in measures for some kinetic models of collective motion}, Math. Mod. Meth. in Appl. Sci., {\bf 21}, (2011), 515--539.

\bibitem{CCR2} J. A. Ca\~nizo, J. A. Carrillo, J. Rosado, \textit{Collective Behavior of Animals: Swarming and Complex Patterns}, Arbor {\bf 186}, (2010), 1035--1049.

\bibitem{CCH} J. A. Carrillo, Y.-P. Choi, and M. Hauray, \textit{The derivation of Swarming models: Mean-field limit and Wasserstein distances}, in Collective dynamics from Bacteria to Crowds: An Excursion Through Modeling, Analysis and Simulation, Springer, 2014.

\bibitem{CFRT} J. A. Carrillo, M. Fornasier, J. Rosado, G. Toscani, \textit{Asymptotic Flocking Dynamics for the kinetic Cucker-Smale model}, SIAM J. Math. Anal., {\bf 42}, (2010), 218--236.

\bibitem{CPJ} T. Champion, L. D. Pascale, and P. Juutinen, \textit{The $\infty$-Wasserstein distance: local solutions and existence of optimal transport maps}, SiAM J. Math. Anal., {\bf 40}, (2008), 1--20.

\bibitem{CKFL} I. D. Couzin, J. Krause, N. R. Franks, and S. A. Levin, \textit{Effective leadership and decision-making in animal groups on the move}, Nature, {\bf 433}, (2005), 513--516.

\bibitem{CD} F. Cucker and J.-G. Dong, \textit{Avoiding collisions in flocks}, IEEE Trans. Automatic Control, {\bf 55}, (2010), 1238--1243.

\bibitem{CS1} F. Cucker and S. Smale, \textit{On the mathematics of emergence}, Jpn. J. Math., {\bf 2}, (2007), 197--227.

\bibitem{CS2} F. Cucker and S. Smale, \textit{Emergent behavior in flocks}, IEEE Trans. Automat. Control, {\bf 52}, (2007), 852--862.

\bibitem{DM} P. Degond  and S. Motsch, \textit{Large-scale dynamics of the persistent Turing Walker model of fish behavior}, J. Stat. Phys., {\bf 131}, (2008), 989 -- 1022.

\bibitem{GS} C. R. Givens, and R. M. Shortt, \textit{A class of Wasserstein metrics for probability distributions}, Michigan Math. J., {\bf 31}, (1984), 231--240.

\bibitem{HHK} S.-Y. Ha, T. Ha, J.-H. Kim, \textit{Emergent behavior of a Cucker-Smale type particle model with nonlinear velocity couplings}, IEEE Trans. Auto. Control, {\bf 55}, (2010), 1679--1683.

\bibitem{HL} S.-Y. Ha, and J.-G. Liu, \textit{A simple proof of the Cucker-Smale flocking dynamics and mean-field limit}, Comm. Math. Sci., {\bf 7}, (2009), 297--325.

\bibitem{HT} S.-Y. Ha, and E. Tadmor, \textit{From particle to kinetic and hydrodynamic descriptions of flocking}, Kinetic and Related Models, {\bf 1}, (2008), 415--435.


\bibitem{HJ1} M. Hauray and P.-E. Jabin, \textit{$N$-particles approximation of the Vlasov equations with singular potential}, Arch. Rational Mech. Anal., {\bf 183}, (2007), 489--524.

\bibitem{HJ2} M. Hauray and P.-E. Jabin, \textit{Particles approximations of Vlasov equations with singular forces: Propagation of chaos}, arXiv:1107.3821, 2013.

\bibitem{LP} N. E. Leonard, D. A. Paley, F. Lekien, R. Sepulchre, D. M. Fratantoni, and R. E. Davis, \textit{Collectivemotion,sensor
networks and ocean sampling}, Proc. IEEE, {\bf 95}, (2007), 48--74.

\bibitem{MC} R. J. McCann, \textit{Stable rotating binary stars and fluid in a tube}, Houston J. Math., {\bf 32}, (2006), 603--631.

\bibitem{PKH} J. Park, H. Kim, and S.-Y. Ha, \textit{Cucker-Smale flocking with inter-particle bonding forces}, IEEE Tran. Automatic Control, {\bf 55}, (2010), 2617--2623.

\bibitem{PE} J. Parrish and L. Edelstein-Keshet, \textit{Complexity, pattern, and evolutionary trade-offs in animal aggregation}, Science, {\bf 294}, (1999), 99--101.

\bibitem{PEG} L. Perea, P. Elosegui and G. G\'omez, \textit{Extension of the Cucker-Smale control law to space flight formation}, J. Guidance, Control and Dynamics, {\bf 32}, (2009), 526--536.

\bibitem{Pes} J. Peszek, \textit{Existence of piecewise weak solutions of a discrete Cucker-Smale's flocking model with a singular communication weight}, arXiv:1302.4224, 2013.

\bibitem{TT} J. Toner and Y. Tu, \textit{Flocks, herds, and schools: a quantitative theory of flocking}, Phys. Rev. E, {\bf 58}, (1998), 4828--4858.

\bibitem{Vic} T. Vicsek, A. Czir\'ok, E. Ben-Jacob, I. Cohen, and O. Schochet, \textit{Novel type of phase transition in a system of
self-driven particles}, Phys. Rev. Lett., {\bf 75}, (1995), 1226--1229.

\bibitem{Vil} C. Villani, \textit{Topics in optimal transportation}, volume 58 of Graduate Studies in Mathematics, American Mathematical Society, Providence, RI, 2003.

\end{thebibliography}
\end{document}